\documentclass[journal]{IEEEtran}
\usepackage[pdftex]{graphicx}
\usepackage{wrapfig}
\usepackage{subfigure}
\usepackage[symbol]{footmisc}
\usepackage{bm}
\usepackage{bbm}
\usepackage{amsmath}
\usepackage{cases}
\usepackage{color}
\usepackage{enumitem}

\usepackage{amssymb}

\usepackage{cite}

\ifCLASSINFOpdf
\else
\fi
\usepackage{amsmath,amssymb,amsthm}

\begin{document}

%
\title{Shift-enabled graphs: Graphs where shift-invariant filters are representable as polynomials of shift operations}
%
%
%

\author{Liyan~Chen,
Samuel~Cheng$^\ast$,~\IEEEmembership{Senior~Member,~IEEE,}
                Vladimir~Stankovic,~\IEEEmembership{Senior~Member,~IEEE,}
        and~Lina~Stankovic,~\IEEEmembership{Senior~Member,~IEEE}
\thanks{L. Chen is with the Department
of Computer Science and Technology, Tongji University, Shanghai,
 201804 China and the School of Mathematics, Physics \& Information Science, Zhejiang Ocean University, Zhejiang, 316022, China (e-mail: chenliyan@tongji.edu.cn).}
\thanks{S. Cheng is with the Department of Computer Science and Technology, Tongji University, Shanghai, 201804 China and the School of Electrical and Computer Engineering, University of Oklahoma, OK 74105, USA (email: samuel.cheng@ou.edu).}
\thanks{V. Stankovic and L. Stankovic are with Department of Electronic and Electrical Engineering,
         University of Strathclyde, Glasgow, G1 1XW U.K.
        (e-mail:\{vladimir.stankovic,~lina.stankovic\}@strath.ac.uk).}
\thanks{$^\ast$ Corresponding author.}
\thanks{
}
}

%
%

\markboth{}%
{Shell \MakeLowercase{\textit{et al.}}: Bare Demo of IEEEtran.cls for IEEE Journals}
%



\newtheorem{lemma}{Lemma}

\maketitle

\begin{abstract}
In digital signal processing, a shift-invariant filter can be represented as a polynomial expansion of a shift operation, that is, the $Z$-transform representation.
When extended to Graph Signal Processing (GSP),
this would mean that a shift-invariant graph filter can be represented as a polynomial of the shift matrix of the graph.
Prior work shows that this holds under the shift-enabled condition that the characteristic and minimum polynomials of the shift matrix are identical.
While the shift-enabled condition is often ignored in the literature, this letter shows that this condition is essential for the following reasons. 
First, we prove that this condition is not just sufficient but also \textbf{\emph{necessary}} for any shift-invariant filter to be representable by the shift matrix.
Moreover, we provide a counterexample showing that given a filter that commutes with a non-shift-enabled graph, it is generally impossible to convert the graph to a shift-enabled graph with a shift matrix still commuting with the original filter. 
The result provides a deeper understanding of shift-invariant filters when applied in GSP and shows that further investigation of shift-enabled graphs is needed to make them applicable to practical scenarios.
 \end{abstract} %


%
\IEEEpeerreviewmaketitle

\section{Introduction}
In Digital Signal Processing (DSP), a filter $F(\cdot)$ is a system that takes a signal $\alpha$ as an input and generates a new signal $\tilde {\alpha}=F \left( \alpha \right) $ as an output. An important class of digital filters are time- or shift-invariant filters \cite{Proakis_1996_Digital,Rabiner1975Theory,Oppenheim1999Discrete}, for which the following property holds~\cite{sandryhaila_2014_big_data}:
\begin{equation}
 F_1\left(F_2\left(\alpha \right) \right)=F_2\left(F_1\left( \alpha\right) \right),\label{H1H2}
\end{equation}
which guarantees that, for two filters $F_1(\cdot)$ and $F_2(\cdot)$, the order of the filters does not change the output.
A consequence of \eqref{H1H2} is that a shift-invariant filter can be represented as a polynomial expansion of shift operation, that is, the $Z$-transform representation of the polynomials in $z^{-1}$
~\cite{ Oppenheim1999Discrete, sandryhaila_2014_big_data,puschel2006algebraic}:
\begin{equation}
F\left( z^{-1}\right) =\sum_{n=-\infty}^{+\infty}f_nz^{-n},\label{hz}
\end{equation}
for some coefficients $f_n$ that define the Z-transform.

GSP extends classic DSP to signals with inherent structures by combining algebraic and spectral graph theory with DSP ~\cite{Puschel2008Algebraic,sandryhaila_2013_discrete}. Many DSP concepts can be extended to {\it graph signals}, including frequency analysis, signal convolution, and filtering~\cite{Shuman_2013_The_emerging_field,Girault2015Translation,Ortega2017Graph}. 

Of particular interest to this study is whether a graph filter can be represented as a polynomial of the shift operator matrix $S$ of the graph \cite{Shuman_2013_The_emerging_field}.
If a graph filter $H$ can be represented as a polynomial in $S$,
then filter output can be computed efficiently under any data structure implementation since the output signal at any node only involves inputs of its immediate neighborhood.  Consequently, filtering a signal by $H^k$ can be achieved by repeating the above procedure $k$ times.
Thus, the computation complexity is reduced. Therefore, it is not just theoretically interesting, but is also practically useful to know when a filter can be represented as a polynomial of the shift matrix.

More formally, let $S$ be an $N\times N$ shift matrix. 
By Cayley-Hamilton Theorem, any polynomial $h(S)$ of shift matrix $S$ can be written as a polynomial with degree less than $N$. Thus, the degree of freedom{\footnote{The degree of freedom of a matrix is referred to as the number of independent elements in this matrix.}} of $h(S)$ is at most $N$. Yet, in the matrix representation, a filter matrix $H$ has the degree of freedom $N \times N$, and therefore most filters cannot be represented as polynomials of $S$.
On the other hand, when $H$ can be represented as a polynomial of $S$, it is obvious that $H$ and $S$ commute, i.e., $H$ is a shift-invariant graph filter~\cite{sandryhaila_2013_discrete}. So the interesting question is when is shift-invariant filter $H$  representable as a polynomial of $S$.

For the above problem, \cite{sandryhaila_2013_discrete} provided a sufficient condition that the characteristic polynomial $p_S(\lambda)$ and minimal polynomial $m_S(\lambda)$ of the shift matrix $S$ are identical
(i.e., $p_S(\lambda)=m_S(\lambda)$, which will be referred to as shift-enabled condition; see Definition~\ref{Def1}).
However, it also argued in \cite{sandryhaila_2013_discrete} that this condition could be disregarded in practice (see Theorem 2 in~\cite{sandryhaila_2013_discrete} and the discussion thereafter).
In this letter, we argue that the shift-enabled condition should not be ignored through a concrete counterexample and rigorous analysis.
Current literature avoids the issue by either assuming that the shift operator must be a normal matrix \cite{Marques_2017}, or considering symmetric graphs where the adjacency matrix is diagonalizable \cite{Girault2015Translation} (which often cannot be guaranteed since it comes from the physical problem in hand), 
or imposing the explicit constraint that a filter must be a polynomial of the shift operator 
\cite{Shuman_2013_The_emerging_field} or focusing only on the case when shift matrix has distinct eigenvalues, which implies $p_S(\lambda)=m_S(\lambda)$ \cite{Ortega2017Graph}. 

The main contributions of the paper are as follows:
\begin{itemize}
        \item 
        Contrary to current literature, we illustrate that the shift-enabled condition is essential and should not be overlooked in practice.
        \item
        We tighten up the shift-enabled condition showing that it is not just a sufficient, but also a {\it necessary}, condition for any shift-invariant filter $H$ to be represented as a polynomial in shift matrix $S$.
        \item 
        We provide a counterexample showing that given a filter $H$ that commutes with a non-shift-enabled matrix $S$, 
        it is not always
        possible to design a shift-enabled matrix $\tilde{S}$ such that we can represent $H=h(\tilde{S})$ for some polynomial $h(\cdot)$.     
\end{itemize}

In Section \uppercase\expandafter{\romannumeral2}, the basic concepts of GSP are discussed.
Section \uppercase\expandafter{\romannumeral3} provides the theoretical guarantee that the shift-enabled condition is {\it necessary} for
a shift-invariant graph filter to be representable as a polynomial in $S$. Section \uppercase\expandafter{\romannumeral4} presents a concrete counterexample.
Conclusions and a possible extension are discussed in Section \uppercase\expandafter{\romannumeral5}.



\section{Basic concepts of graph signal processing}

In this section, we briefly review notations and concepts of GSP that are relevant to our study. For more details, see\cite{sandryhaila_2014_big_data,Shuman_2013_The_emerging_field,sandryhaila_2013_discrete,sandryhaila_2014_discrete_frequency,chen_2015_sampling_theory}.

GSP studies signals on graphs, where a graph ${\mathcal G}=({\mathcal V},A)$ is determined by its set of vertices
 ${\mathcal V} =(v_0,v_1,\cdots,v_{N-1}) $ and its adjacency matrix $A$, which reflects the relationship between vertices, such as similarity or dependency.
 A graph signal $\bm{x} :{\mathcal V}\mapsto\mathbb{C}$ is a mapping from a set of vertices to the complex field and can be expressed conveniently as a vector $\bm{x}:$
\begin{equation*}
\bm{x}=\left(x_0,x_1,\cdots,x_{N-1} \right)^T,
\end{equation*}
where for each $i$, $x_i$ is indexed by a node $v_i \in {\mathcal V}$.


For example, for a cyclic directed graph signal $\bm{x}$:
\begin{equation*}
\tilde{\bm{x}}=A\bm{x}=(x_{N-1},x_0,\cdots,x_{N-2})^T,
\end{equation*}
shifts $x_{n-1}$ to the next signal $\tilde{x}_n=x_{(n-1)}\text{mod} N$. 
Hence, the matrix $A$ is sometimes also treated as a shift matrix or {\em shift operation}~\cite{sandryhaila_2013_discrete,sandryhaila_2014_discrete_frequency,sandryhaila_2014_big_data}. The shift matrix can be replaced by other matrices which reflect the relation of vertices, such as
the Laplacian matrix, the normalized Laplacian matrix, the diffusion matrix, and so on \cite{Marques_2017}. In the remaining letter we use $S$ to denote the general shift matrix.

With the graph shift operation defined, the shift-invariant filters can be naturally extended from classic DSP to GSP, with the shift matrix $S$ in place of the shift operation $z^{-1}$. Indeed, a graph signal filter $H(\cdot)$ is shift-invariant if and only if the filter $H(\cdot)$ and shift matrix $S$ satisfy:
\begin{equation}
S(Hx)=H(Sx).
\end{equation}
This means that applying a graph shift to the filtered signal is equivalent to a graph filter applied to the shifted signal\cite{sandryhaila_2013_discrete}.


Note that in classic DSP, Eq.~\eqref{hz} is an immediate consequence of shift-invariance, and hence holds for all shift-invariant filters.
Just as in DSP, a shift-invariant filter $H(\cdot)$ in GSP can be written as a polynomial in $S$, {\it but under the condition $p_S(\lambda)=m_S(\lambda)$}, thus, there exists a shift-invariant graph filter that cannot be represented as a polynomial in $S$, as  discussed in the next section.

\section{The sufficient and necessary condition for shift-invariant filters to be polynomials of the shift operator}

In contrast to~\cite{sandryhaila_2013_discrete}, 
we emphasize that $p_S(\lambda)=m_S(\lambda)$ is not just sufficient but also {\em necessary}
for any shift-invariant graph filter to be representable as a polynomial in $S$.

For convenience, we introduce the notion of a ``shift-enabled" graph as follows:
\newtheorem{myDef}{Definition}
\begin{myDef}\label{Def1}
	A graph ${\mathcal G}$ is shift-enabled 
	{\footnote {Note that this definition differs from \cite{Grady}. In \cite{Grady} a graph is called shift-invariant if there exists a permutation of the node ordering such that the Laplacian matrix representing the graph is circulant.}}
		if the characteristic and minimal polynomials of its corresponding shift matrix $S$ are equivalent, i.e., $p_S(\lambda)=m_S(\lambda)$. 
	We also say that $S$ is shift-enabled when the above condition is satisfied. \end{myDef}

If the shift-enabled condition holds for matrix $S$, we have the following theorem.
\newtheorem{Thm}{Theorem}
\newtheorem{remark}{Remark}
\begin{Thm}\label{Theorem1} 
	The matrix $S$ is shift-enabled if and only if every matrix $H$ commuting with $S$ is a polynomial in $S$. That is, we can write
    \begin{equation}\label{H_poly}
        H=h(S)=h_0I+h_1S+\cdots+h_LS^L ~(h_l\in\mathbb{C}, L\in\mathbb{N}),
        \end{equation}
                                where $I$ as the identity matrix.
\end{Thm}
The major difference between Theorem~\ref{Theorem1} in our letter and Theorem~1 in~\cite{sandryhaila_2013_discrete} is the {\it necessity} of the shift-enabled condition proved in Appendix A. 
Thus, if $p_S(\lambda)\neq m_S(\lambda)$, then there must exist a filter $H$ that cannot be represented as a polynomial in $S$, even if $HS=SH$. 
However, the shift-enabled condition has been widely ignored in the literature, because of the following theorem from \cite{sandryhaila_2013_discrete}.

\newtheorem{Thm 2}{Theorem 2\cite{sandryhaila_2013_discrete}}
\begin{Thm}\label{Thm2}
        For any matrix $S$, there exists a matrix $\tilde{S}$ and matrix polynomial $r(\cdot)$, such that $S=r(\tilde{S})$ and $\tilde{S}$ is shift-enabled~\rm\cite{sandryhaila_2013_discrete}.
\end{Thm}

The proof can be found in \cite{sandryhaila_2013_discrete}.
According to Theorem~\ref{Thm2} in~\cite{sandryhaila_2013_discrete}, it is argued\footnote{We note that the authors focused on adjacency matrix $A$, but the concept can easily be generalized.} that for any graph filter $H$, 
we can construct a composite function $g=h\circ r$ as a polynomial such that $h(S)=g(\tilde S)$, thus the condition $p_S(\lambda)=m_S(\lambda)$ can be ignored since $h(S)$ and $g(\tilde S)$ are now equivalent and $p_{\tilde S}(\lambda)=m_{\tilde S}(\lambda)$ holds. 
While this argument appeared to be promising, we believe that it is unfortunately misleading.
First, the argument seems to ignore that $h(\cdot)$ is not necessarily a polynomial to begin with. Otherwise, if $h(\cdot)$ was already a polynomial, converting $S$ to $\tilde{S}$ would be unnecessary.
Furthermore, Theorem~\ref{Thm2} cannot guarantee that the newly constructed $\tilde{S}$ commutes with $H$ anymore. And so in general, we cannot ensure that $H$ is a polynomial in $\tilde{S}$.


In the following section, we provide a concrete example illustrating that we should not bypass the shift-enabled condition using Theorem 2. It is the intention of the authors that this example sheds some light on how common intuition in DSP cannot always be applied to GSP.

%
%
%
%


%
%

\section{Examples of shift-invariant filters for non-shift-enabled graph not representable as polynomial of converted shift-enabled graph}

\subsection{Example filter}
Consider the special shift matrix
$S=\begin{pmatrix}
\begin{smallmatrix}
0 & 1 & 1 \\
0 & 0 & 0 \\
0 & 0 & 0 \\
\end{smallmatrix}
\end{pmatrix}.
~$ Note that
$p_S\left( \lambda \right)=\lambda^3\neq \lambda^2=m_S\left( \lambda\right)$ 
and hence $S$ is not shift-enabled. Since, as mentioned earlier, this is a necessary (not just merely a sufficient) condition, there exists a  shift-invariant filter not representable as a polynomial of $S$.

Let
\vspace{5pt}
$H=\begin{pmatrix}
\begin{smallmatrix}
0 & 1 & 0\\
0 & 0 & 0 \\
0 & 0 & 0 \\
\end{smallmatrix}
\end{pmatrix}
$ be a filter.
Note that $HS=\bm{0}=SH$ and thus the filter is shift-invariant. 
Yet, it is impossible to find polynomial representation of $H$ in terms of $S$.
Indeed, note that $(S^n)_{1,2}=(S^n)_{1,3}$, $\forall n\in{\{0,1,2,\cdots\}}$. Hence, we must have $(h(S))_{1,2}=(h(S))_{1,3}$ for any polynomial $h(S)$. But since $H_{1,2}=1\neq 0=H_{1,3}$, $H\neq h(S)$ for any polynomial function $h(\cdot)$.
\subsection{Conversion of non-shift-enabled graphs to shift-enabled graphs}
Ref.~\cite{sandryhaila_2013_discrete} suggested that shift-enabled condition can be ignored as we could convert the graph into a shift-enabled graph through a polynomial transformation. But as we will show in the following, any ``transformed" $S$ will not commute with the filter $H$, which makes the transformation pointless. 


Assume $S$ can be written as a polynomial of a shift-enabled  $\tilde{S}\triangleq
\begin{pmatrix}
\begin{smallmatrix}
\tilde{s}_{11} & \tilde{s}_{12} & \tilde{s}_{13}\\
\tilde{s}_{21} & \tilde{s}_{22} & \tilde{s}_{23}\\
\tilde{s}_{31} & \tilde{s}_{32}& \tilde{s}_{33}\\
\end{smallmatrix}
\end{pmatrix}$ as stated in Theorem 2. If we also assume that $H$ is shift-invariant under $\tilde S$ as before ($H\tilde{S}=\tilde{S}H$), we have
$\tilde{S}=
\begin{pmatrix}
        \begin{smallmatrix}
                \tilde{s}_{11} & \tilde{s}_{12} &\tilde{s}_{13}\\
                0 & \tilde{s}_{22} & 0\\
                0 & \tilde{s}_{32}& \tilde{s}_{33}\\
\end{smallmatrix}
\end{pmatrix}$ and $\tilde{s}_{11}=\tilde{s}_{22}$.
But then this will contradict with either  $\tilde {S}$ being shift-enabled or $S$ representable as a polynomial of $\tilde S$ as shown below. 

There are two possible cases.
\begin{enumerate}
\item  If $\tilde{s}_{11}=\tilde{s}_{22}=\tilde{s}_{33}$ then
$\tilde{S}=
\begin{pmatrix}
\begin{smallmatrix}
\tilde{s}_{11} & \tilde{s}_{12} &\tilde{s}_{13}\\
0 & \tilde{s}_{11} & 0\\
0& \tilde{s}_{32}& \tilde{s}_{11}\\
\end{smallmatrix}
\end{pmatrix}$
and we have  $p_{\tilde{S}}(\lambda)=(\lambda-\tilde{s}_{11})^3$. Note that
$(\tilde{S}-\tilde{s}_{11}I)^2=
\begin{pmatrix}
        \begin{smallmatrix} 0 & \tilde{s}_{13}\tilde{s}_{32} & 0 \\
        0 & 0 & 0\\
        0 & 0 & 0
        \end{smallmatrix}
\end{pmatrix}$ and  consequently we must have $\tilde{s}_{13}\tilde{s}_{32}\neq0$ because, otherwise, this\  will break the shift-enabled condition as 
 $p_{\tilde S}(\lambda)=(\lambda-\tilde{s}_{11})^3\neq(\lambda-\tilde{s}_{11})^2 =m_{\tilde S}(\lambda)$.

On the other hand, $S$ cannot be represented as polynomial in $\tilde{S}$ when $\tilde{s}_{13}\tilde{s}_{32}\neq0$. Following Cayley-Hamilton Theorem, 
$S=r(\tilde{S})=r_0I+r_1\tilde{S}+r_2{\tilde{S}}^2$. From there we obtain two equations:
\begin{equation}{\label{case1_1}}
\begin{cases}
\tilde{s}_{13}(r_1+2r_2\tilde{s}_{11})=1,\\
\tilde{s}_{32}(r_1+2r_2\tilde{s}_{11})=0.\\
\end{cases}
\end{equation}
One can easily conclude that $\tilde{s}_{32}=0$, but  this contradicts $\tilde{s}_{13}\tilde{s}_{32}\neq0$.
Therefore, when  $\tilde{s}_{11}=\tilde{s}_{22}=\tilde{s}_{33}$, we cannot find an $\tilde{S}$ which satisfies $S=r(\tilde{S})$, $p_{\tilde{S}}(\lambda)=m_{\tilde{S}}(\lambda)$ and $H\tilde{S}=\tilde{S}H$ at the same time.

\item  If $\tilde{s}_{11}=\tilde{s}_{22} \neq \tilde{s}_{33}$, 
we will also obtain conflicting results. Assume that
$S=r(\tilde{S})=r_0I+r_1\tilde{S}+r_2{\tilde{S}}^2$, then we obtain  three equations that cannot simultaneously hold:

\begin{subnumcases}{}
r_0+r_1\tilde{s}_{11}+r_2{\tilde{s}_{11}}^2=0, \label{eqn:7a}\\
r_0+r_1\tilde{s}_{33}+r_2{\tilde{s}_{33}}^2=0, \label{eqn:7b} \\
\tilde{s}_{13}[r_1+r_2(\tilde{s}_{11}+\tilde{s}_{33})]=1. \label{eqn:7c}
\end{subnumcases}\\

By combining the results of Eq.~\eqref{eqn:7a} minus Eq.~\eqref{eqn:7b} and $\tilde{s}_{11}\neq \tilde{s}_{33}$, we have $r_1+r_2(\tilde{s}_{11}+\tilde{s}_{33})=0$ which contradict with Eq.~\eqref{eqn:7c}.
Therefore, $S=r(\tilde{S})$ and $H\tilde{S}=\tilde{S}H$ cannot simultaneously hold  when $\tilde{s}_{11}\neq \tilde{s}_{33}$.
\end{enumerate}

As shown in the the above example, while we can find a new $\tilde S$ such that $S=r(\tilde{S})$ and $p_{\tilde S}(\lambda) = m_{\tilde S}(\lambda)$, any new $\tilde S$ no longer commutes with the original filter $H$ and thus it is impossible to use Theorem~\ref{Thm2} to prove that $H$ is representable by $\tilde S$.
On the other hand, note that in Theorem~\ref{Theorem1}, $H\tilde{S} = \tilde{S}H$ is only a sufficient condition for $H$ to be representable by $\tilde S$ (given that $\tilde S$ is shift-enabled).
So potentially, there may still exist $\tilde S$ such that $S=r(\tilde S)$ and $H=h(\tilde S)$ (but $H\tilde{S} = \tilde{S}H$ is not satisfied). However, we have evaluated our example using the  symbolic math toolbox of MATLAB and concluded that such $\tilde S$ (regardless if $H\tilde{S} = \tilde{S}H$) does not exist as well.

In summary, with the condition given by Theorem~\ref{Thm2} ($S=r(\tilde S)$), the $H$ in our example cannot be represented as polynomial of any such $\tilde S$.
That is, there exists a graph ${\mathcal G}$ associated with a non-shift-enabled shift matrix $S$ and a corresponding shift-invariant filter $H$. Yet, it is impossible to represent $H$ as a polynomial of $S$ or $\tilde{S}$ regardless of the conversion procedure suggested by Theorem~\ref{Thm2}. 
\subsection{A class of filters}
The above example can be extended to the following class of filters:
\begin{equation*}
\mathbb{H}=\{\alpha H+ q(S)|\alpha \in \mathbb{R}, q(S) \mbox{ is a polynomial of $S$}\} \label{a_class_of_H},
\end{equation*}
where $S$ and $H$ are as defined in the previous subsection.

Since apparently $q(S)S=Sq(S)$ for any polynomial $q(S)$ and $HS=SH$ as discussed above, any filter $H'\triangleq\alpha H + q(S)$ commutes with $S$ as well.
Thus any filter in $\mathbb{H}$ is shift-invariant. 
However, since $H$ is not representable as a polynomial of transformed matrix $\tilde{S}$, so does 
$H'$. Otherwise, since $q(S)=q\circ r(\tilde{S})$, $H=\frac{1}{\alpha}(H'-q\circ r(\tilde S))$ is also a polynomial in $\tilde{S}$ and this contradicts with what we have shown earlier.
Therefore, when the shift-enabled condition is not satisfied, we may in fact find an infinite number of shift-invariant filters that violated the argument drawn from Theorem~\ref{Thm2}.

\section{Conclusion}


This letter emphasizes the importance and {\it necessity} of the shift-enabled condition.
Then, we present a counterexample showing that given a filter $H$ that commutes with a non-shift-enabled matrix $S$, it is generally impossible to convert $S$ to a shift-enabled matrix $\tilde{S}$ such that we can represent $H=h(\tilde{S})$ for some polynomial $h(\cdot)$.
This 
provides a deeper understanding of shift-invariant filters under the GSP umbrella. In fact, we conjecture that $p_S(\lambda)=m_S(\lambda)$ may have deeper implications, and these corresponding  shift-enabled graphs that demonstrate enhanced properties of shift-invariant filters may have distinct characteristics and structures apart from graphs that do not satisfy the condition.

An apparent future direction is to study rules and structures that may be used to identify the shift-enabled graphs. Moreover, it is also interesting to see if one may decompose non-shift-enabled graphs into shift-enabled subgraphs, to optimize the design of the GSP filters.


%

\appendices
\section{} 
%
We will only include here the proof that the shift-enabled condition is necessary for any shift-invariant matrix to be representable as a polynomial of the shift matrix. Please refer to~\cite{sandryhaila_2013_discrete} for the proof of sufficiency.


\newtheorem{prop}{Proposition}
\begin{prop}\rm{(The necessity of shift-enabled condition.)}\label{prop_Jordon_poly}
	If every matrix $H$ commuting with $S$ is a polynomial in $S$ then $S$ is shift-enabled ($p_S(\lambda)=m_S(\lambda)$).
\end{prop}
Proposition~\ref{prop_Jordon_poly} shows the necessity of shift-enabled condition. Combining with sufficiency in~\cite{sandryhaila_2013_discrete}, shift-enabled condition is not only sufficient but also {\it necessary} condition.
\begin{lemma}\label{Lemma1_Jordan_eigen}
        A graph shift matrix $S$ is shift-enabled if and only if each Jordan block in the Jordan canonical form of $S$ is associated with a distinct eigenvalue \rm{(see Proposition 6.6.2 in \cite{lancaster_1985_matrix_theory}. Proof: omitted.)}.
\end{lemma}

Moreover, all matrices commuting with the Jordan matrix satisfy conditions described in the following lemma.

\begin{lemma}\label{Lemma2_cummute_Jordan}
Let $S= PJ P^{-1}$ with $J=diag[J_1,\cdots,J_K]$ being its Jordan normal form, where $J_k$ is an $n_k\times n_k$ Jordan block associated with eigenvalue $\lambda_k$, for $k=1,2,\cdots, K$.
Then $SX=XS$ if and only if $X=PY P^{-1}$, where $Y$ is a block matrix whose blocks with dimensions match with $J$. That is, the $(l,m)$-th block  $Y^{(l,m)}$ has size $n_l\times n_m$. Moreover, $Y^{(l,m)}=0$ if $\lambda_m \neq \lambda_l$. Otherwise, $Y^{(l,m)}$ has the following forms: 
$Y^{(l,m)}=\begin{bmatrix}\begin{smallmatrix} 0 & U\end{smallmatrix}\end{bmatrix}$ for $n_l\leqslant n_m$ and $Y^{(l,m)}=\begin{bmatrix} \begin{smallmatrix} U\\ 0\\ \end{smallmatrix} \end{bmatrix}$
for $n_l> n_m$, where $U$ is a square upper-triangular Toeplitz matrix \rm{(Theorem 12.4.1 of~\cite{lancaster_1985_matrix_theory})}.
\end{lemma}
\begin{proof}
As $S=PJP^{-1}$ and $SX=XS$, then $JY=YJ$, where $Y=P^{-1}XP$. Since $J$ is a block diagonal matrix, this reduces to
\begin{equation}\label{eqn:JY=YJ}
J_lY^{(l,m)}=Y^{(l,m)}J_m.
\end{equation}
Note that $J_k=\lambda_k I_{n_k} + N_{n_k}$ 
is a $n_k\times n_k$ matrix with constant 1 at upper-diagonal and 0 everywhere else. Then, Eq.~\eqref{eqn:JY=YJ} is equal to
\begin{equation}
(\lambda_l -\lambda_m) Y^{(l,m)}=Y^{(l,m)} N_{n_m} - N_{n_l}Y^{(l,m)}.
\label{eqn:jordan_commute1}
\end{equation}
Consider two possible cases.
\begin{enumerate}
\item If $\lambda_l\neq\lambda_m$, then multiple $(\lambda_l - \lambda_m)$ on both sides of Eq.~\eqref{eqn:jordan_commute1} and reapply Eq.~\eqref{eqn:jordan_commute1}. We will have
	   \begin{flalign*}
	    \begin{split}
	       &~~~(\lambda_l-\lambda_m)^2 Y^{(l,m)}\\
	       & =Y^{(l,m)}(N_{n_m})^2-2N_{n_l}Y^{(l,m)}N_{n_m}+(N_{n_l})^2Y^{(l,m)}.
	     \end{split}
	    \end{flalign*}
	Repeat this step iteratively and ultimately for any positive integer $t$, we will have
	\begin{align*}(\lambda_l\!-\!\lambda_m)^t Y^{(l,m)}\!=\!\sum_{i=0}^t (\!-\!1)^i {\dbinom{t}{i}} (N_{n_l})^i Y^{(l,m)} (N_{n_m})^{t-i}.
	 \end{align*}

	Note that $(N_{n_k})^{n_k}=0, $ for $k=1,\cdots,K $. Therefore, for sufficiently large $t$, we have either $(N_{n_l})^i=0$ or $(N_{n_m})^{t-i}=0$. Combining with $\lambda_l\neq \lambda_m$, we have $Y^{(l,m)}=0$.
	\item If $\lambda_l=\lambda_m$, then Eq.~\eqref{eqn:jordan_commute1} becomes $Y^{(l,m)} N_{n_m}= N_{n_l} Y^{(l,m)}$.  Recall that the dimension of $Y^{(l,m)}$ is $n_l \times n_m$. Let us first assume that $n_l \leqslant n_m$, so $Y^{(l,m)}$ is a ``fat" matrix.  Note that $Y^{(l,m)} N_{n_m}$ is equivalent to shifting $Y^{(l,m)}$ to the right by one column (with all zeros filling the first column). And similarly, $N_{n_l} Y^{(l,m)}$ is equivalent to shifting $Y^{(l,m)}$ up by one row (with all zeros filled in the last row).  So comparing the two matrices for the last row, we immediately have $Y^{(l,m)}_{n_l,j}=0$ for $j=1,\cdots,n_m-1$.  Let $Y^{(l,m)}_{n_l,n_m}=\delta_1$. And now compare again\ for the second last row, we have $Y^{(l,m)}_{n_l-1,j}=0$ for $j=1,\cdots, n_m-2$ and $Y^{(l,m)}_{n_l-1,m_l-1}= \delta_1$.  Again, we have freedom to choose the last element, this time let $Y^{(l,m)}_{n_l-1,m_l}=\delta_2$, and repeat the comparison for the third last row and so on. This shows that $Y^{(l,m)}=[0_{n_m-n_l}\ U]$, where $U$ is a upper triangular Toeplitz matrix with diagonal $\delta_1$, first superdiagonal $\delta_2$, second superdiagonal $\delta_3$ and so on.  We can prove the other case in a similar manner when $Y^{(l,m)}$ is thin (i.e., $n_l > n_m$).
\end{enumerate}
\end{proof}
We will use the above lemmas to prove the Proposition~\ref{prop_Jordon_poly}.
\begin{proof}
Let $S$ be an $N\times N$ matrix and have the same form as in Lemma~\ref{Lemma2_cummute_Jordan}.
	Assume that for all $H$, if $HS=SH$, then $H=h(S)$ for some polynomial $h(\cdot)$. By Cayley-Hamilton, it is sufficient to consider $h(\cdot)$ of degree at most $N-1$. Thus the degree of freedom of $H$ is at most $N$. However, from Lemma~\ref{Lemma2_cummute_Jordan},
	the degree of freedom of $H$ to satisfies $HS=SH$ is always larger than  $n_1 + n_2 + \cdots + n_K=N$ unless all Jordan blocks have distinct eigenvalues. Thus from Lemma~\ref{Lemma1_Jordan_eigen}, $S$ is shift-enabled.
\end{proof}

\section*{Acknowledgment}

The authors thank M. Ye, B. Zhao and D. Jakovetic for helpful discussions. This project has received funding from  the European Union’s Horizon 2020 research and innovation programme under the Marie Sklodowska-Curie grant agreement no. 734331 and  the Fundamental Research Funds for the Central Universities no. 0800219369.

\ifCLASSOPTIONcaptionsoff
  \newpage
\fi




\bibliographystyle{IEEEtran}
\bibliography{reference}

\end{document}